\documentclass{article}
\usepackage{graphicx}
\usepackage{microtype}
\usepackage{graphicx}
\usepackage{amsmath}
\usepackage{amsthm}
\usepackage{amsfonts}
\usepackage{amssymb}
\usepackage{clrscode3e}
\usepackage{bm}
\usepackage{upgreek}
\usepackage{tikz}
\usepackage{tikzscale}

\usepackage{tikz}
\usepackage{tikz-qtree}

\title{An Efficient Sampling Algorithm for Difficult Tree Pairs} 

\author{Sean Cleary and Roland Maio}

\begin{document}

\maketitle

\begin{abstract}
It is an open question whether there exists a polynomial-time algorithm for computing the rotation distances between pairs of extended ordered binary trees.
The problem of computing the rotation distance between an arbitrary pair of trees, $(S, T)$, can be efficiently reduced to the problem of computing the rotation distance between a difficult pair of trees $(S', T')$, where there is no known first step which is guarranteed to be the beginning of a minimal length path. Of interest, therefore, is how to sample such difficult pairs of trees of a fixed size. We show that it is possible to do so efficiently, and present such an algorithm that runs in time $O(n^4)$.
 \end{abstract}

\newtheorem{definition}{Definition}
\newtheorem{theorem}{Theorem}
\newtheorem{corollary}{Corollary}
\newtheorem{lemma}{Lemma}
\theoremstyle{plain}
\newtheorem{proposition}{Proposition}
\newcommand{\ds}{\displaystyle}

\section{Introduction}
Trees are a fundamental data structure with wide applications ranging from efficient search (such as binary search trees) to modelling biological processes (such as phylogenetic trees). We routinely are interested in calculating some metric of interest between trees.

 One widely-considered tree distance metric on trees with a natural left-to-right order on leaves is that of the rotation distance between a pair of extended ordered binary trees. There are no known polynomial-time algorithms for computing rotation distance.  Culik and Wood \cite{culik1982note} described rotation distance. Sleator, Tarjan and Thurston \cite{slt88} used the correspondence between trees with $n$ internal nodes and triangulations of the marked regular $n+2$-gon to show that if there is a common edge between the two triangulations then any shortest path does not flip this edge.  Such a common edge thus breaks the rotation distance problem into two smaller sub-problems.  Furthermore, they  showed that if it is possible to flip an edge of either polygon to obtain a common edge, then there is a shortest path which begins by doing so.  We call edges which are not common but which can be flipped to become a common edge {\em one-off} edges, as they are one move away from being common edges themselves.  Cleary and St.~John  \cite{rotfpt} used these reduction rules  to show that rotation distance is fixed parameter tractable.

 We call a pair of trees with no common edges and no edges which can be immediately flipped to create a common edge a {\em difficult tree pair}.   The above reductions transform the problem of computing the rotation distance on a pair of trees drawn from all possible pairs to a pair of trees drawn from the set of all such difficult tree pairs. A common edge, arising either immediately or from a performing a single flip to change a one-off edge to a common edge, then naturally splits the tree pair into a pair of smaller tree pairs, as explained in Sleator, Tarjan, and Thurston \cite{slt88}.  The kernel of the difficult of the rotation distance problem at this point is to find distances between difficult pairs.  
 
To understand how effective different approximation and partial algorithms are at evaluating and estimating rotation distance, it would be useful to sample difficult tree pairs. 
 It is possible to find examples of difficult tree pairs by picking a tree pair of large size at random, and then performing all possible reductions and one-off moves, splitting the problem into a collection of smaller subproblems, until either the trees are identical (extremely unlikely)  or until a collection of difficult tree pairs is obtained.  But such a procedure is not only time-consuming, it is not possible to tell in advance how many reductions there will be and what the resulting sizes of the smaller remaining difficult piece pairs will be.  Thus there is no control on the resulting size of the difficult tree pairs produced.   In general (see Cleary, Rechnitzer and Wong \cite{randomf}) there are a sizable number of common edges and one-off edges, resulting on average about at least a 10\% reduction in the size of a randomly selected tree pair to a largest difficult remaining tree pair. 
 It is not difficult to construct specific examples of specified size of difficult tree pairs- examples of Dehornoy \cite{dehornoy}, Pournin \cite{pournin}, and Cleary and Maio \cite{badconflicts} are families of difficult pairs but in each case of a restricted type.  In many of these very specific cases, analysis to that family of instances can give coincident upper and lower bounds on rotation distance, giving an exact calculation.  But these families are very sparse in the set of all difficult tree pairs.   The set of all difficult tree pairs appears to grow exponentially with size, but at a slower expontial growth rate than the set of all tree pairs, per work of Cleary and Maio \cite{counthard} suggesting that the fraction of all tree pairs decreases exponentially at a rate of about $0.77^n$, with already ratio of less than 1 in a billion tree pairs of size 70 being difficult and the fraction dropping with further increases of size.
 
  Difficult tree pairs lie at the kernel of a number of questions of interest.  Because the rotation distance problem frequently splits into smaller subproblems, the essential difficulties are contained in the set of difficult tree pairs.  Difficult tree pairs can be used to test estimation algorithms for rotation distance, to find estimates for typical rotation distance between tree pairs selected at random, and look for difficult pathological behavior for rotation distance paths.
  
  This motivates studying difficult tree pairs in their own right. We describe below an efficient algorithm for sampling difficult tree pairs of a specified size.  This sampling is not uniform across all difficult tree pairs of a prescribed size but does have wide coverage of such pairs.

The algorithm we describe can be seen as a variation on Remy's algorithm \cite{remy} for efficiently generating rooted ordered trees uniformly at random, but instead of working on growing the size of a single tree, we grow a pair of trees while applying a filtering criterion.  Unlike Remy's algorithm, the difficult pairs are not sampled uniformly at random but having an efficient (polynomial-time) means of generating pairs is useful for understanding rotation distance problem instances better and for testing the performance of new algorithms.  Computational experiments show that the distribution of selected tree pairs are not uniformly random but there does seem to be wide dispersion, with broad coverage of difficult tree pairs.

\begin{figure}
    \centering
     \includegraphics[width=\textwidth]{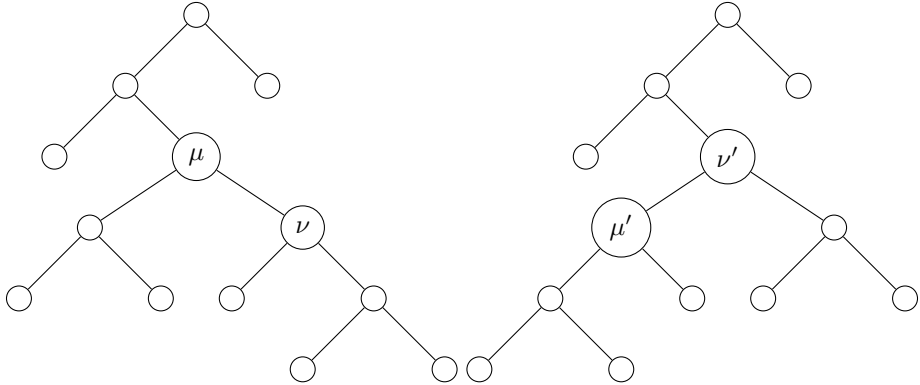}
    \caption{An example of rotation at a node, with a rotation promoting node $\nu$ to node $\nu'$.}
    \label{figrotation}
\end{figure}

\section{Background}
An {\em extended ordered binary tree} is a rooted binary tree where every node has exactly $0$ or $2$ children and whose leaves are labelled starting with $0$ in their order defined by a pre-order traversal from the root. The label of a leaf node $l$ is denoted \func{label}$(l)$. The size of an extended ordered binary tree $T$, denoted $|T|$, is the number of internal nodes $T$ contains. The set of all nodes in $T$ is denoted \func{nodes}($T$). In the following {\em tree} will refer to an extended ordered binary tree and $S$ and $T$ will be trees of the same size.

A rotation in a tree is a local operation which promotes an internal node $\nu$ to the position of its parent $\mu$, demotes $\mu$ to one of $\nu$'s children, and makes one of $\nu$'s children a child of $\mu$, illustrated in Figure \ref{figrotation}. We will denote the partial function that returns the parent of a node by $\pi : \func{nodes}(T) \mapsto \func{nodes}(T)$, with $\pi(\func{root}(T))$ is undefined. 
We adopt the convention that the statement ``rotate at a node $\nu$'', denoted \func{rotate}($\nu$), means to perform that rotation which promotes $\nu$ to the position of its parent. For a tree of size $n$ there are $n - 1$ possible rotations, one for each internal node excepting the root.

 Given a pair of trees $(S, T)$ of the same size, it is possible to transform one into the other by some sequences of rotations. The minimum length of any such sequence defines the {\em rotation distance} between $S$ and $T$, which we denote $d(S, T)$.

The interval of a node $\nu$, \func{interval}$(\nu)$, is the pair $(\alpha, \beta)$ where $\alpha$ is the label of the least-labelled leaf in the tree rooted at $\nu$ and $\beta$ is the label of the greatest-labelled leaf in the tree rooted at $\nu$. The label $\alpha$ is called the {\em lower bound of the interval of $\nu$} and is denoted $\lfloor\func{interval}(\nu)\rfloor$. Similarly, the label $\beta$ is called the {\em upper bound of the interval of $\nu$} and is denoted $\lceil\func{interval}(\nu)\rceil$. If $\nu$ is a leaf, then its lower bound is the same as its upper bound and is defined to be its label. If $\nu$ is an internal node, then its lower bound is the lower bound of its left child, and its upper bound is the upper bound of its right child; formally
\[
    \func{interval}(\nu) =
    \begin{cases}
        (\func{label}(\nu), \func{label}(\nu))   \hfill &\text{if $\nu$ is a leaf}\\
        \hfill (\lfloor\func{interval}(\func{left}(\nu))\rfloor ,\lceil\func{interval}(\func{right}(\nu))\rceil) \hfill &\text{otherwise}
    \end{cases}
\]

The intervals of a tree $T$, denoted \func{intervals}($T$), is the set of all the intervals of the internal nodes of $T$.

The labels $\alpha$ and $\beta$ are related to each other by the size of the subtree rooted at $\nu$ in the following way:

\begin{proposition}\label{betaFromAlpha}
    Let $\nu$ be an internal node of $T$, and $N$ the subtree rooted at $\nu$, and $(\alpha, \beta) = \func{interval}(\nu)$, then $\beta = \alpha + |N|$.
\end{proposition}

\begin{proof}
    Recall that $N$ has $|N|+1$ leaves. It is a property of pre-order traversal that once the traversal visits a node it will visit the entire subtree rooted at that node before it visits any other part of the tree. Consequently, when the pre-order traversal reaches $\nu$, the next $|N|+1$ leaf nodes that will be visited will be the leaf nodes of $N$. Thus, the greatest label any leaf in $N$ can have is $\alpha + |N|$ and this must be attained by the last leaf that is visited in $N$.
\end{proof}

In addition to changing one tree into another, a rotation in a tree $T$ at a node $\nu$ has the effect of replacing one of the intervals of the tree by a new one. This new interval is uniquely determined by $T$ and $\nu$ and is denoted $\func{1-interval}(\nu)$. The $\func{1-interval}(\nu)$ can be defined in terms of the intervals of $\nu$, the parent of $\nu$, and the children of $\nu$. If $\nu$ is the left child of its parent, then the lower bound of $\func{1-interval}(\nu)$ is the lower bound of the right child of $\nu$ and the upper bound of $\func{1-interval}(\nu)$ is the upper bound of the parent of $\nu$. If $\nu$ is the right child of its parent, then the lower bound of $\func{1-interval}(\nu)$ is the lower bound of its parent, and the upper bound of $\func{1-interval}(\nu)$ is the upper bound of the left child of $\nu$. Formally
\[
    \func{1-interval}(\nu) =
    \begin{cases}
        (\lfloor\func{interval}(t)\rfloor,\lceil\func{interval}(\pi(\nu))\rceil)    \hfill &\text{if $\nu = \func{left}(\pi(\nu))$}\\
        (\lfloor\func{interval}(\pi(\nu))\rfloor,\lceil\func{interval}(s)\rceil) \hfill &\text{otherwise}
    \end{cases}
\]
where $s = \func{left}(\nu)$ and $t = \func{right}(\nu)$.

 The \func{1-intervals}($T$) is the set of all $n-1$ intervals that can be obtained by rotating some node in $T$.

Trees correspond naturally to the marked triangulations of a polygon, we denote the corresponding triangulation by $\bigtriangleup (T)$. The edges of $\bigtriangleup(T)$ correspond to the intervals($T$).

While the reduction rules were first developed from the perspective of triangulations of the polygon, they may be formulated from the tree perspective in terms of intervals and rotations. A common edge between triangulations corresponds to a common interval occuring in the intervals of both trees. A one-off edge between triangulations corresponds to a common interval that can be obtained by rotating at one of the nodes in $S$ or $T$.

The binary word of $T$, \func{word}($T$), is obtained by beginning with the empty string, traversing $T$ in pre-order and appending at each node a `1' if the node is an internal node and a `0' otherwise. Thus the symbol at the $i$th index in \func{word}($T$) is determined by the $i$th node visited in $T$ by a pre-order traversal. This determines a mapping from symbols in \func{word}($T$) to \func{nodes}($T$).

\begin{definition}
    Let $T$ be an extended ordered binary tree, and let $\nu$ be the $i$-th node visited in a pre-order traversal of $T$. The symbol of $\nu$ in $\func{word}(T)$, denoted $\func{sym_T}(\nu)$, is defined to be the the symbol of $\func{word}(T)$ at index $i$.
\end{definition}

The following property of \func{word}$(T)$ 
gives one method for computing the intervals of $T$.

\begin{proposition}\label{labelFromWord}
    Let $l$ be a leaf node of $T$, then the label of $l$ is given by the number of `0's that precede $\func{sym_T}(l)$.
\end{proposition}

\begin{proof}
    Suppose the label of $l$ is $\alpha$. By the definition of label, $l$ is the $(\alpha + 1)$st leaf node visited in the preorder traversal of $T$. In computing $\func{word}(T)$, therefore, exactly $\alpha$ `0's must have been appended before $\func{sym_T}(l)$ is appended.
\end{proof}

\begin{theorem}
    \label{upperFromSize}
    Let $T$ be an extended ordered binary tree, $\nu$ an internal node of $T$, $N$ the subtree of $T$ rooted at $\nu$, and $(\alpha, \beta) = \func{interval}(\nu)$. Then $\alpha$ is given by the number of 0's that precede $\func{sym_T}(\nu)$, and $\beta = \alpha + |N|$.
\end{theorem}

\begin{proof}
    To prove $\alpha$ is given by the number of 0's that precede $\func{sym_T}(\nu)$ it suffices, by Proposition \ref{labelFromWord}, to show that the symbol of the leaf node, $l$, with label $\alpha$ is the first 0 that proceeds $\func{sym_T}(\nu)$. Suppose this is not the case, then there is at least one $0$ that proceeds $\func{sym_T}(\nu)$ and precedes $\func{sym_T}(l)$. Then there is some leaf node $k$ in the subtree rooted at $\nu$ that is visited after $\nu$ and before $l$. So the label of $k$ is at most $\alpha-1$, but this contradicts the assumption that $l$ is the least labelled leaf. Finally, from Proposition \ref{betaFromAlpha} it follows that $\beta = \alpha + |N|$.
\end{proof}

We let $\circ$ denote string concatenation and we let $\nu$ be a node of a tree. We define the functions \func{left}($\nu$) and \func{right}($\nu$) to return the left or right child or $\nu$ respectively. A recursive definition for \func{word}($\nu$) can then be given as follows
\[
    \func{word}(\nu) =
    \begin{cases}
         \hfill 1\circ\func{word}(\func{left}(\nu))\circ\func{word}(\func{right}(\nu))   \hfill & \text{if $\nu$ is an internal node}\\
         0                               \hfill & \text{otherwise}
    \end{cases}
\]

With this definition $\func{word}(T) = \func{word}(r)$ where $r$ is the root of $T$.

Remy's algorithm \cite{remy} is a method for sampling trees of a fixed size uniformly at random by growing a tree larger at each stage ensuring that each possible tree of that size is equally likely to be generated. The algorithm begins with a tree of size 1 and iteratively grows the tree until a tree of the desired size is obtained. On each iteration, one of the internal or external nodes, say $\nu$, of the current tree, say $T$, is selected uniformly at random. Then a new node, $\mu$, is created. The new node $\mu$ takes the place of $\nu$ in the tree, and $\nu$ is set as the left or right child of $\mu$ with equal probability. We say that the resulting tree is obtained from $T$ by growing left (or right) at $\nu$.
\\
\indent If a tree $S$ may be grown in some way by an iteration of Remy's algorithm to obtain a tree $T$, then we call $T$ a growth neighbor of $S$ and denote the set of all growth neighbors of $S$ by \func{growthNeighbors}($S$).

On an iteration of Remy's algorithm, if an external node is chosen to be grown, then growing left or right will result in the same tree. Thus, an upper bound on the number of growth neighbors a tree of size $n$ may have is $3n + 1$.

\section{Difficult Pair Sampling Algorithm}
The difficult pair sampling algorithm, DPS, begins by randomly choosing one of the 4 difficult pairs of trees of size 4. We call these difficult pairs primitive because there are no difficult pairs of trees that are smaller. The algorithm then iteratively grows the pair of trees in size by 1 until a pair of the desired size is obtained. On each iteration, for the current pair of trees $S$ and $T$, DPS finds all difficult pairs of trees ($U$, $V$) such that $U$ is a growth neighbor of $S$ and $V$ is a growth neighbor of $V$ and randomly selects one of these pairs to be the next $S$ and $T$.

\begin{codebox}
\Procname{$\proc{DPS}(n)$}
\li $S$, $T \gets \func{randomPrimitiveDifficultPair}()$
\li \For $i \gets 5$ \To $n$
\li     \Do
            $\id{choices} \gets \emptyset$
\li         \For $U$ \kw{in} $\func{growthNeighbors}(S)$
\li             \Do
                    \For $V$ \kw{in} $\func{growthNeighbors}(T)$
\li                     \Do
                            \If $\func{isDifficultPair}(U, V)$
\li                             \Then $\attrib{\id{choices}}{add}((U, V))$
                            \End
                        \End
                \End
\li         $S$, $T \gets \attrib{\id{choices}}{randomElement}()$
        \End
\li \Return $S$, $T$
\end{codebox}

What is not obvious about DPS is that for an arbitrary difficult pair $(S, T)$, it is always possible to grow $S$ and $T$ into a difficult pair $(U, V)$. We will show that this is the case by examining a particular growth neighbor- there may be additional ones but a single one suffices for proving the correctness of the algorithm.

\begin{definition}
    Let $T$ be an extended ordered binary tree of size $n$, and $\omega$ be the internal node of $T$ whose right child is the leaf with label $n$. The extended ordered binary tree of size $n+1$, obtained by growing $T$ at $\omega$ left, will be denoted $\sigma(T)$.
\end{definition}

We will show that given a difficult pair $(S, T)$, the pair of trees $(\sigma(S), \sigma(T))$ is also a difficult pair. The proof that $(\sigma(S), \sigma(T))$ is a hard pair will rest on the relation between \func{intervals}($T$) to \func{intervals}($\sigma(T)$) and \func{1-intervals}($T$) to \func{1-intervals}($\sigma(T)$).

Relating \func{intervals}($T$) to \func{intervals}($\sigma(T)$) and \func{1-intervals}($T$) to \func{1-intervals}($\sigma(T)$) will require relating \func{word}($T$) to \func{word}($\sigma(T)$) which we will do next.

\begin{lemma}
    \label{word2word}
    Let $T$ be an extended ordered binary tree of size $n$, then $\func{word}(T) = \Lambda\Omega$ and $\func{word}(\sigma(T)) = \Lambda 1\Omega 0$
\end{lemma}

\begin{proof}
    Let $\omega$ be the parent of the leaf node with label $n$ in $T$, this implies that $\omega$, and all of its ancestors are either the root or the right child of their parent. From the definition of \func{word} it follows that $\func{word}(T)$ is of the form $\Lambda\Omega$ where $\Omega = \func{word}(\omega)$.\\
    \indent When $T$ is grown left at $\omega$, the new node, $\phi$, will take $\omega$ as its left child and become the right child of $\omega$'s former parent. Consequently, $\func{word}(\sigma(T))$ will be $\Lambda\Phi$ where $\Phi = \func{word}(\phi)$. 
    \begin{align*}
        \func{word}(\phi) &= 1\circ\func{word}(\func{left}(\phi))\circ\func{word}(\func{right}(\phi))\\
                          &= 1\circ\func{word}(\omega)\circ 0\\
                     \Phi &= 1\Omega 0
    \end{align*}
\end{proof}

\begin{figure}
    \centering
    \includegraphics[width=4in]{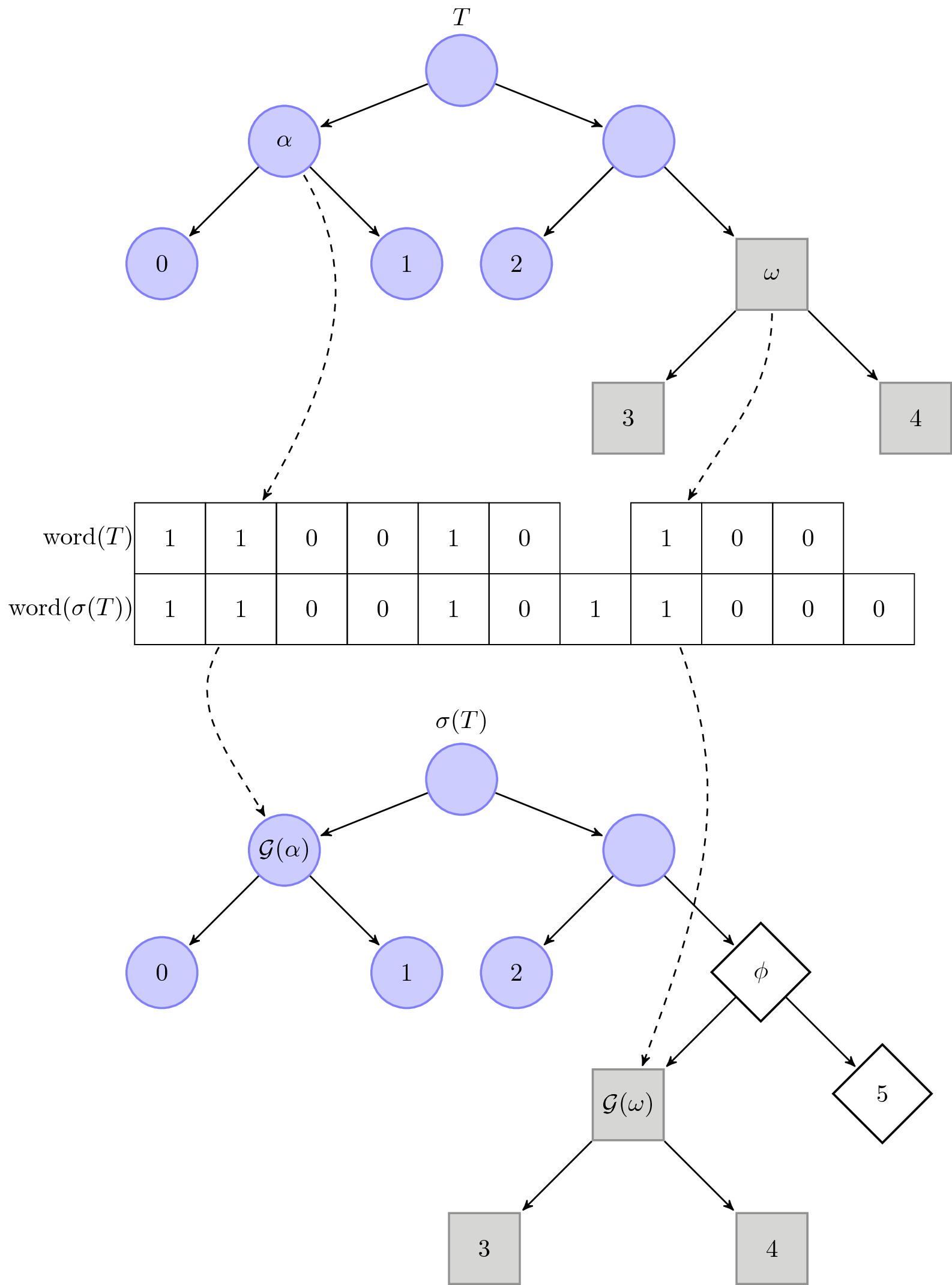}
    \caption{Growing a node and preserving the difficulty of the pair.}
    \label{fig:nodeGrowth}
\end{figure}

There is a natural way in which the nodes in $T$ correspond to the nodes in $\sigma(T)$. The intuition for this is shown in Figure \ref{fig:nodeGrowth}. This correspondence can be formalized in terms of the relation between \func{word}$(T)$ and \func{word}$(\sigma(T))$.

\begin{definition}
    Let $T$ be an extended ordered binary tree. The natural growth injection of the nodes of $T$ to the nodes of $\sigma(T)$, $\mathcal{G}: \func{nodes}(T) \mapsto \func{nodes}(\sigma(T))$ is defined as
    
    $\mathcal{G}(\nu) = \func{sym_{\sigma(T)}^{-1}}(\func{word}(\sigma(T))(i + \bm{1}\{|\Lambda| < i\}))$
    
    where $i$ is the index of $\func{sym}_T(\nu)$ and $\bm{1\{ \ldots \}}$ is the indicator function.
\end{definition}

There are several properties of $\mathcal{G}$ which will be critical to proving our claim that DPS can always grow a difficult pair $(S, T)$ into another difficult pair $(U, V)$. The first that we will examine relates the interval of a node, $\nu$, in $T$ to the interval of $\mathcal{G}(\nu)$ in $\sigma(T)$.

\begin{lemma}
    \label{interval2interval}
    Let $T$ be an extended ordered binary tree of size $n$, $\omega$ be the node of $T$ whose right child is the leaf with label $n$, $\nu$ be any node of $T$ that is not $\omega$ and $(\alpha, \beta) = \func{interval}(\nu)$. Then $\func{interval}(\mathcal{G}(\nu)) = (\alpha, \beta + \bm{1}\{\beta = n\})$.
\end{lemma}

\begin{proof}
    Let $(\gamma, n) = \func{interval}(\omega)$. By Lemma \ref{word2word} we have $\func{word}(T) = \Lambda\Omega$ and $\func{word}(\sigma(T)) = \Lambda 1\Omega 0$. Now we partition $\func{nodes}(T) - \{\omega\}$ into three sets:
    \begin{enumerate}
        \item $\{\nu : (\alpha, \beta) = \func{interval}(\nu), \alpha < \gamma, \beta < n\}$ all nodes that have an interval with a lower bound less than the lower bound of $\omega$ and with an upper bound that is less than $n$.
        \item $\{\nu : (\alpha, \beta) = \func{interval}(\nu), \alpha < \gamma, \beta = n\}$ all nodes that have an interval with a lower bound less than the lower bound of $\omega$ and with an upper bound equal to $n$.
        \item $\{\nu : (\alpha, \beta) = \func{interval}(\nu), \gamma \leq \alpha\}$ all nodes that have an interval with a lower bound greater than or equal to the lower bound of $\omega$.
    \end{enumerate}
    
    We consider the first case. Let $N$ be the subtree rooted at $\nu$, $a$ be the leaf with label $\alpha$, $b$ be the leaf with label $\beta$ and $g$ be the leaf with label $\gamma$. Since $\alpha < \gamma$, the symbol of $a$ in $\func{word}(T)$ must precede the symbol of $g$ in $\func{word}(T)$ and therefore $\func{sym}_T(a) \in \Lambda$ and $\func{sym}_T(\nu) \in \Lambda$. By the definition of $\mathcal{G}$ it follows that the index of $\func{sym}_{\sigma(T)}(\mathcal{G}(\nu))$ is the same as the index of $\func{sym}_T(\nu)$ and so $\func{sym}_{\sigma(T)}(\mathcal{G}(\nu)) \in \Lambda$. Therefore, by Proposition \ref{labelFromWord} the lower bound of $\mathcal{G}(\nu)$ must be $\alpha$. Since $\beta < n$ it follows $\beta < \gamma$, otherwise $g$ is in $N$, but this would imply that $\omega$ is also in $N$ and so $\beta = n$, which is impossible. Consequently the symbol of $b$ in $\func{word}(T)$ must precede the symbol of $g$ in $\func{word}(T)$ and therefore $\func{sym}_T(b) \in \Lambda$ and $\func{sym}_{\sigma(T)}(\mathcal{G}(b)) \in \Lambda$. This implies that $\func{word}(\nu) = \func{word}(\mathcal{G}(\nu))$ and so the size of the subtree rooted at $\mathcal{G}(\nu)$ is the same as the size of $N$. Applying Theorem \ref{upperFromSize} it follows that the upper bound of $\mathcal{G}(\nu)$ is $\beta$. Therefore $\func{interval}(\mathcal{G}(\nu)) = (\alpha, \beta) = (\alpha, \beta + \bm{1}\{\beta = n\})$.
    
    Now we consider the second case. Let $N$ be the subtree rooted at $\nu$, $a$ be the leaf with label $\alpha$, $b$ be the leaf with label $\beta$ and $g$ be the leaf with label $\gamma$. Since $\alpha < \gamma$, the symbol of $a$ in $\func{word}(T)$ must precede the symbol of $g$ in $\func{word}(T)$ and so $\func{sym}_T(a) \in \Lambda$ and $\func{sym}_T(\nu) \in \Lambda$. By the definition of $\mathcal{G}$, the index of $\func{sym}_{\sigma(T)}(\mathcal{G}(\nu))$ is the same as the index of $\func{sym}_T(\nu)$ and by Proposition \ref{labelFromWord} the lower bound of $\mathcal{G}(\nu)$ is $\alpha$. Since $\beta = n$, $\omega$ must be contained in the subtree rooted at $\nu$. Therefore, when $T$ is grown left at $\omega$, the subtree rooted at $\mathcal{G}(\nu)$ will be larger in size by one than $N$. Applying Theorem \ref{upperFromSize} it follows that the upper bound of $\mathcal{G}(\nu)$ is $\beta + 1$. And so $\func{interval}(\mathcal{G}(\nu)) = (\alpha, \beta + 1) = (\alpha, \beta + \bm{1}\{\beta = n\})$.
    
    Finally, we consider the third case. Let $N$ be the subtree rooted at $\nu$. Since $\gamma \leq \alpha$, $\nu$ is a descendant of $\omega$. Therefore $\func{word}(\nu)$ is a proper substring of $\Omega$. Observe that for $|\Lambda | < i \leq |\func{word}(T)|$ we have $\func{word}(T)(i) = \func{word}(\sigma(T))(i+1)$, combined with the definition of $\mathcal{G}$ it follows that $\func{word}(\nu) = \func{word}(\mathcal{G}(\nu))$ so the size of the subtree rooted at $\mathcal{G}(\nu)$ is the same as the size of $N$. Since the number of `0's which precede $\func{sym}_{\sigma(T)}(\mathcal{G}(\nu))$ is the same as the number of `0's that precede $\func{sym}_T(\nu)$ it follows that the lower bound of $\mathcal{G}(\nu)$ is $\alpha$. Therefore $\func{interval}(\mathcal{G}(\nu)) = (\alpha, \beta) =  (\alpha, \beta + \bm{1}\{\beta = n\})$.
    
\end{proof}

The growth injection $\mathcal{G}$ also captures several of the node-to-node relationships of $T$ which are preserved in $\sigma(T)$: $\mathcal{G}$ preserves the relation between parents and their left children, $\mathcal{G}$ preserves the relation between parents and their right children, except for the parent of the node whose right child is the leaf labelled $n$, and taken together, a consequence of the preceding two properties is that $\mathcal{G}$ preserves the relation between parents and children except for the node whose right child is the leaf labelled $n$. We will next state and prove these relationships formally because we will exploit them in proving the relationship between $\func{1-intervals}(T)$ and $\func{1-intervals}(\sigma(T))$.

\begin{lemma}
    \label{left2Gleft}
    Let $T$ be an extended ordered binary tree of size $n$ and $\omega$ be the internal node of $T$ whose right child is the leaf with label $n$. For any internal node $\nu \in T$, the image of the left child of $\nu$ under $\mathcal{G}$ is the left child of the image of $\nu$ under $\mathcal{G}$, that is, $\mathcal{G}(\func{left}(\nu)) = \func{left}(\mathcal{G}(\nu))$.
\end{lemma}

\begin{proof}
    It will suffice to show that the index of the symbol of $\mathcal{G}(\func{left}(\nu))$ in the \func{word} of $\sigma(T)$ is the same as the index of the symbol of $\func{left}(\mathcal{G}(\nu))$ in the \func{word} of $\sigma(T)$. Let $i$ be the index of $\func{sym}_T(\nu)$ and consider two cases as to whether $|\Lambda| < i$ or not.
    
    In the first case, the index of $\func{sym}_T(\func{left}(\nu)) = i + 1$ and so the index of $\func{sym}_{\sigma(T)}(\mathcal{G}(\func{left}(\nu))) = i + 2$. By definition, the index of $\func{sym}_{\sigma(T)}(\mathcal{G}(\nu)) = i + 1$ giving index $\func{sym}_{\sigma(T)}(\func{left}(\mathcal{G}(\nu))) = i + 2$.
    
    In the second case, we must verify that $i+1 \leq |\Lambda|$. Since $\nu$ is an internal node, its symbol must be a `1' and since the suffix of $\Lambda$ is the word of the left child of the parent of $\omega$, the last symbol in $\Lambda$ must be a `0' and so $i \leq |\Lambda| - 1$ which implies $i + 1 \leq |\Lambda|$. Now the index of $\func{sym}_T(\func{left}(\nu)) = i + 1$. Since $i + 1 < |\Lambda|$ the index of $\func{sym}_{\sigma(T)}(\mathcal{G}(\func{left}(\nu))) = i + 1$. By definition, the index of $\func{sym}_{\sigma(T)}(\mathcal{G}(\nu)) = i$ and so the index of $\func{sym}_{\sigma(T)}(\func{left}(\mathcal{G}(\nu))) = i + 1$.
\end{proof}

\begin{lemma}
    \label{right2Gright}
    Let $T$ be an extended ordered binary tree of size $n$ and $\omega$ be the internal node of $T$ whose right child is the leaf with label $n$. For any internal node $\nu \in T$, except for $\pi(\omega)$, the image of the right child of $\nu$ under $\mathcal{G}$ is the right child of the image of $\nu$ under $\mathcal{G}$, that is $\mathcal{G}(\func{right}(\nu)) = \func{right}(\mathcal{G}(\nu))$.
\end{lemma}

\begin{proof}
    It will suffice to show that the index of the symbol of $\mathcal{G}(\func{right}(\nu))$ in the \func{word} of $\sigma(T)$ is the same as the index of the symbol of $\func{right}(\mathcal{G}(\nu))$ in the \func{word} of $\sigma(T)$. Let $i$ be the index of $\func{sym}_T(\nu)$ and consider two cases as to whether $|\Lambda| < i$ or not.
    
    In the first case, the index of $\func{sym}_T(\func{right}(\nu)) = i + |\func{word}(\func{left}(\nu))| + 1$ and so the index of $\func{sym}_{\sigma(T)}(\mathcal{G}(\func{right}(\nu))) = i + |\func{word}(\func{left}(\nu))| + 2$. By definition, the index of $\func{sym}_{\sigma(T)}(\mathcal{G}(\nu)) = i + 1$ and so the index of $\func{sym}_{\sigma(T)}(\func{right}(\mathcal{G}(\nu))) = i + |\func{word}(\func{left}(\mathcal{G}(\nu)))| + 2$. Observe that the upper bound of the interval of $\func{left}(\nu)$ must be less than $n$ since it is the left child of its parent. Applying Lemmas \ref{interval2interval} and \ref{left2Gleft} it follows that $\func{interval}(\func{left}(\nu)) = \func{interval}(\mathcal{G}(\func{left}(\nu))) = \func{interval}(\func{left}(\mathcal{G}(\nu)))$. Therefore the size of the subtree rooted at $\func{left}(\nu)$ is the same as the size of the subtree rooted at $\func{left}(\mathcal{G}(\nu))$ and so $|\func{word}(\func{left}(\nu))| = |\func{word}(\func{left}(\mathcal{G}(\nu)))|$
    
    In the second case, we must verify that $i + |\func{word}(\func{left}(\nu))| + 1 \leq |\Lambda|$. But this is clearly true since the only case in which $|\Lambda| < i + |\func{word}(\func{left}(\nu))| + 1$ is when $\nu = \pi(\omega)$, but we have excluded $\pi(\omega)$ from consideration. Observe that the index of $\func{sym}_T(\func{right}(\nu)) = i + |\func{word}(\func{left}(\nu))| + 1$ and so the index of $\func{sym}_{\sigma(T)}(\mathcal{G}(\func{right}(\nu))) = i + |\func{word}(\func{left}(\nu))| + 1$. By definition, the index of $\func{sym}_{\sigma(T)}(\mathcal{G}(\nu)) = i$ and so the index of $\func{sym}_{\sigma(T)}(\func{right}(\mathcal{G}(\nu))) = i + |\func{word}(\func{left}(\mathcal{G}(\nu)))| + 1$. Since $\func{left}(\nu)$ is a left child the upper bound of its interval must be less than $n$. Applying lemmas 2 and 3 it follows that $\func{interval}(\func{left}(\nu)) = \func{interval}(\mathcal{G}(\func{left}(\nu))) = \func{interval}(\func{left}(\mathcal{G}(\nu)))$. This implies that the size of the subtree rooted at $\func{left}(\nu)$ is the same as the size of the subtree rooted at $\func{left}(\mathcal{G}(\nu))$ and so $|\func{word}(\func{left}(\nu))| = |\func{word}(\func{left}(\mathcal{G}(\nu)))|$.
\end{proof}

Together Lemma \ref{left2Gleft} and Lemma \ref{right2Gright} show the preserved parent structure.

\begin{corollary}
    \label{parent2Gparent}
    Let $T$ be an extended ordered binary tree of size $n$ and $\omega$ be the internal node of $T$ whose right child is the leaf with label $n$. For any node $\nu \in \func{nodes}(T) - \{\omega\}$, the image of the parent of $\nu$ under $\mathcal{G}$ is the parent of the image of $\nu$ under $\mathcal{G}$, that is, $\mathcal{G}(\pi(\nu)) = \pi(\mathcal{G}(\nu))$.
\end{corollary}

The next lemma will relate \func{intervals}($T$) to \func{intervals}($\sigma(T)$).

\begin{lemma}
    \label{intervals2intervals}
    Let $T$ be an extended ordered binary tree of size $n$ and $\omega$ be the internal node of $T$ whose right child is the leaf with label $n$. Then the intervals of $\sigma(T)$ are related to the intervals of $T$ by
    \begin{align}
        \func{intervals}(\sigma(T)) = &\{(\alpha, \beta + \bm{1}\{\beta = n\}) : (\alpha, \beta) \in \func{intervals}(T)\} \cup \{\func{interval}(\omega)\}
    \end{align}
\end{lemma}

\begin{proof}
    Let $(\gamma, n) = \func{interval}(\omega)$. By Lemma \ref{interval2interval} we have the intervals for the internal nodes of $\sigma(T)$ that are the image under $\mathcal{G}$ of some node in $T$, except for $\omega$. This gives us $n-1$ intervals of $\sigma(T)$ and we have only to consider the interval of $\mathcal{G}(\omega)$ and the interval of $\pi(\mathcal{G}(\omega))$. By Lemma \ref{word2word} the number of `0's which precede $\func{sym}_{\sigma(T)}(\mathcal{G}(\omega))$ is the same as the number of `0's which precede $\func{sym}_{\sigma(T)}(\pi(\mathcal{G}(\omega)))$ which is the same as the number of `0's which precede $\func{sym}_T(\omega)$ and so the intervals of $\mathcal{G}(\omega)$ and $\pi(\mathcal{G}(\omega))$ have the same lower bound, namely $\gamma$, which is the lower bound of $\omega$. By construction, the subtree rooted at $\mathcal{G}(\omega)$ has the same size as the subtree rooted at $\omega$. Applying Proposition \ref{betaFromAlpha} it follows that $\func{interval}(\mathcal{G}(\omega)) = \func{interval}(\omega)$. Also by construction, the subtree rooted at $\pi(\mathcal{G}(\omega))$ is greater in size by 1 than the subtree rooted at $\omega$. Applying Proposition \ref{betaFromAlpha} again yields $\func{interval}(\pi(\mathcal{G}(\omega))) = (\gamma, n + 1) = (\alpha, \beta + \bm{1}\{\beta = n\})$.
\end{proof}

With the relationship between $\func{intervals}(\sigma(T))$ and $\func{intervals}(T)$ proven, we can state and prove the relationship between $\func{1-intervals}(\sigma(T))$ and $\func{1-intervals}(T)$. Our proof strategy will be to determine for each internal node of $T$, $\nu$, the local structure that determines the $\func{1-interval}(\nu)$ in $T$. Then we will determine how that local structure maps to the local structure of $\mathcal{G}(\nu)$ in $\sigma(T)$ and use this to compute $\func{1-interval}(\mathcal{G}(\nu))$.

\begin{lemma}
    \label{1intervalsto1intervals}
    Let $T$ be an extended ordered binary tree of size $n$, $\omega$ the internal node of $T$ whose right child is the leaf with label $n$, and $\phi$ the internal node of $\sigma(T)$ that is the parent of $\mathcal{G}(\omega)$. Then the \func{1-intervals} of $\sigma(T)$ are related to the \func{1-intervals} of $T$ by
    \begin{align*}
        \func{1-intervals}(\sigma(T)) = &\{(\alpha, \beta + \bm{1}\{\beta = n\}) : (\alpha, \beta) \in \Theta\}\\
                            \cup &\{\func{1-interval}(\func{left}(\omega)), (n, n+1), \func{1-interval}(\phi)\}\\
    \end{align*}
    where $\Theta = \func{1-intervals}(T) - \{\func{1-interval}(\omega), \func{1-interval}(\func{left}(\omega))\}$.
\end{lemma}

\begin{proof}
    Let $(\gamma, n) = \func{interval}(\omega)$ and apply lemma $\ref{word2word}$ to obtain $\func{word}(T) = \Lambda\Omega$ and $\func{word}(\sigma(T)) = \Lambda 1\Omega 0$. Now partition $\func{nodes}(T)$ into 6 sets:
    \begin{enumerate}
        \item $\{\nu : (\delta, n) = \func{interval}(\nu), \nu \notin \{\omega, \func{root}(T)\}\}$ every node, excluding the root and $\omega$, whose interval has an upperbound of $n$.
        \item $\{\nu : (\delta, \beta) = \func{interval}(\pi(\nu)), \beta < n, \nu = \func{left}(\pi(\nu))\}$ every node whose parent's interval has an upperbound less than $n$ and that is a left child of its parent.
        \item $\{\nu : (\alpha, \delta) = \func{interval}(\pi(\nu)), \delta < n, \nu = \func{right}(\pi(\nu))\}$ every node whose parent's interval has an upperbound less than $n$ and that is a right child of its parent.
        \item $\{\nu : (\delta, n) = \func{interval}(\func{\pi}(\nu)), (\kappa, n) \neq \func{interval}(\nu), \nu \neq \func{left}(\omega)\}$ every node, $\nu$, excluding the left child of $\omega$, such that the upper bound of the interval of $\nu$ is not $n$ and the upper bound of the interval of $\nu$'s parent is $n$.
        \item $\{\omega\}$ the singleton set containing $\omega$.
        \item $\{\func{left}(\omega)\}$ the singleton set containing the left child of $\omega$.
    \end{enumerate}
    
    In the first case, because we have excluded the root, the parent of $\nu$ is a node in $T$ and its right child must be $\nu$. Since the $\func{1-interval}(\nu) = (\alpha, \beta)$ is obtained by taking the lower bound of $\nu$'s parent and the upper bound of the left child of $\nu$ it follows that $\func{interval}(\func{\pi}(\nu)) = (\alpha, n)$ and $\func{interval}(\func{left}(\nu)) = (\delta, \beta)$. Since we have excluded $\omega$, $\alpha$, $\beta$ and $\delta$ must be less than $\gamma$. Applying Lemma \ref{interval2interval}, Lemma \ref{left2Gleft} and Corollary \ref{parent2Gparent} we have $\func{interval}(\mathcal{G}(\nu)) = (\delta, n+1), \func{interval}(\mathcal{G}(\func{\pi}(\nu))) = \func{interval}(\pi(\mathcal{G}(\nu))) = (\alpha, n+1)$ and $\func{interval}(\mathcal{G}(\func{left}(\nu))) = \func{interval}(\func{left}(\mathcal{G}(\nu))) = (\delta, \beta)$. It follows that $\func{1-interval}(\mathcal{G}(\nu)) = (\alpha, \beta) = (\alpha, \beta + \bm{1}\{\beta = n\}) = \func{1-interval}(\nu)$.
    
    In the second case $\func{interval}(\nu) = (\delta, \epsilon)$ and $\func{interval}(\func{right}(\nu)) = (\alpha, \epsilon)$ for some $\epsilon$ and $\alpha$. Since $\beta < n$, applying Lemma \ref{interval2interval}, Lemma \ref{right2Gright} and Corollary \ref{parent2Gparent} yields $\func{interval}(\mathcal{G}(\nu)) = \func{interval}(\nu)$, $\func{interval}(\mathcal{G}(\pi(\nu))) = \func{interval}(\pi(\mathcal{G}(\nu)))$ and $\func{interval}(\func{right}(\mathcal{G}(\nu))) = \func{interval}(\func{right}(\nu))$. Computing \func{1-interval} of $\mathcal{G}(\nu)$ yields $\func{1-interval}(\mathcal{G}(\nu)) = (\alpha, \beta) = (\alpha, \beta + \bm{1}\{\beta = n\}) = \func{1-interval}(\nu)$.
    
    In the third case $\func{interval}(\nu) = (\epsilon, \delta)$ and $\func{interval}(\func{left}(\nu)) = (\epsilon, \beta)$ for some $\epsilon$ and $\beta$. Since $\delta < n$, applying Lemma \ref{interval2interval}, Lemma \ref{left2Gleft} and Corollary \ref{parent2Gparent} yields $\func{interval}(\mathcal{G}(\nu)) = \func{interval}(\nu)$, $\func{interval}(\mathcal{G}(\pi(\nu))) = \func{interval}(\pi(\nu))$ and $\func{interval}(\func{left}(\mathcal{G}(\nu))) = \func{interval}(\func{left}(\nu))$. Therefore, $\func{1-interval}(\mathcal{G}(\nu)) = (\alpha, \beta) = (\alpha, \beta + \bm{1}\{\beta = n\}) = \func{1-interval}(\nu)$.
    
    In the fourth case, because the upperbound of $\nu$ is not $n$, it follows $\nu$ is the left child of its parent, and so $(\delta, \epsilon) = \func{interval}(\nu)$ and $(\alpha, \epsilon) = \func{interval}(\func{right}(\nu))$ for some $\epsilon$ and $\alpha$ such that $\epsilon < n$.  Consequently $\func{1-interval}(\nu) = (\alpha, n)$ for some $\alpha$ and $\func{interval}(\func{right}(\nu)) = (\alpha, \epsilon)$. Applying Lemma \ref{interval2interval}, Lemma \ref{right2Gright} and Corollary \ref{parent2Gparent} we have that $\func{interval}(\mathcal{G}(\nu)) = \func{interval}(\nu)$, $\func{interval}(\pi(\mathcal{G}(\nu))) = (\delta, n + 1)$ and $\func{interval}(\func{right}(\mathcal{G}(\nu))) = \func{interval}(\func{right}(\nu))$ therefore, $\func{1-interval}(\mathcal{G}(\nu)) = (\alpha, n+1) = (\alpha, \beta + \bm{1}\{\beta = n\})$.
    
    In the fifth case $\func{interval}(\mathcal{G}(\omega)) = (\gamma, n)$, $\func{interval}(\pi(\mathcal{G}(\omega))) = (\gamma, n + 1)$ and the right child of $\mathcal{G}(\omega)$ is the leaf with label $n$. Therefore, $\func{1-interval}(\mathcal{G}(\omega)) = (n, n + 1)$.
    
    In the sixth case, by definition we have $\pi(\func{left}(\omega)) = \omega$ and so $\func{interval}(\pi(\func{left}(\omega))) = \func{interval}(\omega) = (\gamma, n)$. Because the right child of $\omega$ is the leaf with label $n$ it follows that $\func{interval}(\func{left}(\omega)) = (\gamma, n-1)$. Therefore $\func{interval}(\func{right}(\func{left}(\omega))) = (\alpha, n-1)$ for some $\alpha$ and $\func{1-interval}(\func{left}(\omega)) = (\alpha, n)$. By Lemma \ref{interval2interval} we have $\func{interval}(\mathcal{G}(\func{left}(\omega))) = \func{interval}(\func{left}(\omega))$. By Lemma \ref{intervals2intervals} and Corollary \ref{parent2Gparent} we have $\func{interval}(\mathcal{G}(\omega)) = \func{interval}(\omega)$. By Lemma \ref{interval2interval} and Lemma \ref{left2Gleft},  $\func{interval}(\func{right}(\mathcal{G}(\func{left}(\omega)))) = \func{interval}(\func{right}(\func{left}(\omega)))$ so $\func{1-interval}(\mathcal{G}(\func{left}(\omega))) = \func{1-interval}(\func{left}(\omega))$.
    
    The preceding case analysis computes the $\func{1-interval}$ for every internal node of $\sigma(T)$ that is the image under $\mathcal{G}$ of some node $\nu$ in $T$. To complete the $\func{1-intervals}(\sigma(T))$ and the proof, we add $\func{1-interval}(\phi)$.
\end{proof}

With these substitution rules for obtaining the $\func{intervals}$ and $\func{1-intervals}$ of a tree $\sigma(T)$ from the $\func{intervals}$ and $\func{1-intervals}$ of $T$, we can now proceed to show that if $(S, T)$ is a difficult pair, then so too is $(\sigma(S), \sigma(T))$. We will do so by first showing that the pair $(\sigma(S), \sigma(T))$ do not have a common interval between them and secondly that they have no one-off intervals between them either.

\begin{lemma}
    \label{nocommonintervals}
    Let $(S, T)$ be a difficult pair of extended ordered binary trees of size $n$, let $\omega_S$ be the internal node of $S$ whose right child is the leaf with label $n$, and let $\omega_T$ be the internal node of $T$ whose right child is the leaf with label $n$. Then the pair of trees $(\sigma(S), \sigma(T))$ do not have an interval in common.
\end{lemma}

\begin{proof}
    Assume, to the contrary, that the pair $(\sigma(S), \sigma(T))$ have an interval in common. Then some interval, call it $t$, in $\func{intervals}(\sigma(T))$ is also in $\func{intervals}(\sigma(S))$. Consider the possible forms of $t$ given by Lemma \ref{intervals2intervals}. If $t = \func{interval}(\omega_T)$, then $t = (\alpha, n)$ and by Lemma \ref{intervals2intervals}, $\func{interval}(\omega_S) = t$ which implies $(S, T)$ is not a difficult pair. Otherwise, $t = (\alpha, \delta)$, is some other interval in $\func{intervals}(\sigma(T))$. If $\delta = n + 1$, then the interval $(\alpha, n)$ is in both $\func{intervals}(S)$ and $\func{intervals}(T)$ which implies $(S, T)$ is not a difficult pair. Otherwise $\delta < n$ and so $(\alpha, \delta)$ is in $\func{intervals}(S)$ and $\func{intervals}(T)$ and $(S, T)$ is not a difficult pair.
\end{proof}

\begin{lemma}
    \label{nooneoffs}
    Let $(S, T)$ be a difficult pair of extended ordered binary trees of size $n$, let $\omega_S$ be the internal node of $S$ whose right child is the leaf with label $n$, and let $\omega_T$ be the internal node of $T$ whose right child is the leaf with label $n$. Then the pair of trees $(\sigma(S), \sigma(T))$ have no one-off intervals between them.
\end{lemma}

\begin{proof}
    Assume, to the contrary, that the pair $(\sigma(S), \sigma(T))$ have some one-off interval between them. Without loss of generality, let $t$ be the $\func{1-interval}$ of $\sigma(T)$ that is also an interval of $\sigma(S)$ and consider the possible forms of $t$ given by Lemma \ref{1intervalsto1intervals}. By construction, neither $\sigma(S)$ nor $\sigma(T)$ can have the interval $(n, n+1)$ and so $t \neq (n, n+1)$. If $t = \func{1-interval}(\func{left}(\omega_T))$, then $t = (\lfloor\func{right}(\func{left}(\omega_T))\rfloor, n) \in \func{intervals}(\sigma(S))$ but then Lemma \ref{intervals2intervals} implies $t = \func{interval}(\omega_S)$ and so $(S, T)$ is not a difficult pair. If $t = \func{1-interval}(\phi)$ then $t = (\lfloor\pi(\phi_T)\rfloor,n) = (\lfloor\pi(\omega_T)\rfloor, n) = \func{interval}(\pi(\omega_T))$ but then Lemma \ref{intervals2intervals} again implies $t = \func{interval}(\omega_S)$ and so $(S, T)$ is not a difficult pair. If $t = (\alpha, n+1)$, then $(\alpha, n) \in \func{1-intervals}(T)$ and $(\alpha, n) \in \func{intervals}(S)$ and $(S, T)$ is not a difficult pair. Otherwise $t = (\alpha, \delta)$ such that $\delta < n$, therefore $(\alpha, \delta) \in \func{1-intervals}(T)$ and $(\alpha, \delta) \in \func{intervals}(S)$ and so $(S, T)$ is not a difficult pair.
\end{proof}

We have now established the fact that the pair $(\sigma(S), \sigma(T))$ is a difficult tree pair which underlies the correctness of DPS.

\begin{theorem}
    \label{sSsTisdifficult}
    Let $(S, T)$ be a difficult pair of extended ordered binary trees of size $n$, then the pair $(\sigma(S), \sigma(T))$ of extended ordered binary trees of size $n+1$ is a difficult pair.
\end{theorem}

\begin{proof}
    Immediate from Lemma \ref{nocommonintervals} and Lemma \ref{nooneoffs}
\end{proof}

\begin{theorem}
    \label{dpsiscorrect}
    Let $n$ be a natural number greater or equal to 4, then Difficult Pair Sampling algorithm is guaranteed to return a difficult pair of trees of size $n$.
\end{theorem}

\begin{proof}
    We proceed by induction on $n$ the size of the trees in the difficult pair desired. In the base case, $n = 4$, in which case, DPS samples one of the 4 primitive difficult pairs of trees which have been found by enumeration.
    
    Now suppose Difficult Pair Sampling is guaranteed to return a difficult pair of trees for all $m < n$ such that $m,n\in\mathbb{N}$, $4 < n$, and let $(S, T)$ be a difficult pair of trees of size $n-1$ sampled by DPS. By Theorem \ref{sSsTisdifficult} there is at least one pair of difficult trees in the set of all pairs of growth neighbors of $S$ and $T$ which DPS will find by enumeration. Consequently, DPS is guaranteed to return a difficult pair of trees of size $n$.
\end{proof}

\section{Time Complexity of DPS}

We will now analyze the time complexity of DPS and show that it is $O(n^4)$. Line 1 can be implemented to run in constant time by using a table of the primitive difficult pairs. By returning a pair of pointers 
line 9 can also be implemented to run in constant time. Therefore the time complexity of DPS is determined by the {\tt for} loop of lines 2 through 8. We name the {\tt for} loops as follows: let $f$ be the for loop of lines 2 through 8, $g$ be the {\tt for} loop of lines 4 through 7, and $h$ be the {\tt for} loop of lines 5 through 7.

We consider one iteration of $f$ with $(S, T)$ being the current difficult pair and suppose we use a table, $t$, to hold the pairs of difficult growth neighbors of $(S, T)$. Given $n$, we  bound the maximum size of $t$ by the space required for the maximum number of pairs of difficult growth neighbors of size $n$ and so preallocate the space. The space complexity of $t$ is $O(n^3)$. If, on each iteration of $f$, the candidate difficult growth neighbor pairs are stored from the beginning of the table contiguously, then line 3 can be implemented to run in constant time by starting again at the beginning of the table and keeping track of how many rows have been filled. Further, with such a scheme, line 8 can be implemented to run in $O(n)$ time by randomly selecting a row (in constant time) of one of the candidate difficult pairs and then copying the selected candidates (in linear time) to the space allocated for the current pair. The time complexity of one iteration of $f$ is therefore $O(n + O(g))$.

The time complexity of one iteration of $g$ is the sum of the time required to compute one growth neighbor of $S$ and the time complexity of $h$. Using the word representation of an extended ordered binary tree, it is possible to compute a growth neighbor, including its intervals and 1-intervals, in linear time. Thu one iteration of $g$  takes $O(n + O(h))$ steps.

The time complexity of one iteration of $h$ is the sum of the time complexity of computing one growth neighbor of $T$, the time complexity of checking if the resulting pair of growth neighbors, $(U, V)$, is difficult, and the time complexity of adding $(U, V)$ to $choices$. Now suppose we allocate two, two-dimensional tables, $a$ and $b$, where we use $a$ to store $\func{intervals}(S)\cup\func{1-intervals}(T)$ and $b$ to store $\func{intervals}(T)\cup\func{1-intervals}(S)$. These tables will require at most $O(n^2)$ space. If we populate these tables as we construct the growth neighbors, then we can determine whether the pair $(U, V)$ is difficult as we construct $V$ and set a flag appropriately. Then line 6 can be implemented to run in constant time by checking the flag. Assuming the candidate pairs are being stored in table $t$, then adding the pair $(U, V)$ to $choices$ can be a constant time increment operation. So the time complexity of one iteration of $h$ is $O(n)$.

The number of iterations of $h$ is determined by the number of growth neighbors of $T$. As previously mentioned, for a tree of size $n$, a straightforward upper bound on the number of growth neighbors is $3n + 1$. Hence, $h$ will execute $O(n)$ times and $O(h) = O(n^2)$. The same reasoning shows that $g$ will also execute $O(n)$ times and so $O(g) = O(n^3)$. Finally, it is clear that $f$ also executes $O(n)$ times and so $O(f) = O(n^4) = O(\func{DPS})$.

\section{Sampling coverage of DPS}

One of the excellent features of Remy's algorithm for generating trees is that it selects a given tree uniformly at random from all trees.   This DPS algorithm does not have such uniformity, with some difficult tree pairs being sampled more often than others.  Since the number of difficult tree pairs of a particular size is not known exactly, the degree of non-uniformity is difficult to calcultate exactly. From work of Cleary, Elder, Rechnitzer and Taback \cite{randomf}, the fraction of difficult pairs (and in fact its superset, the set of reduced tree pairs) goes to zero exponentially quickly as the size of the tree pairs increase.  Calculations by Cleary and Maio \cite{counthard} show that the number of difficult pairs appear to grow exponentially with an exponential growth rate of about 2.17975, out of the set of equivalence classes of all pairs with an exponential growth rate of about 2.4420.  This is consistent with the observation that for large $n$, the change of selecting a difficult tree pair at random is vanishingly small.

As far as the degree of coverage, for small $n$ where feasible, we found that the DPS algorithm does sample from all hard cases.  The number of distinct difficult pairs for larger $n$  appears to grow exponentially.  This is a starkly broader class of difficult pairs than those specific known earlier examples of Dehornoy \cite{dehornoy}, Pournin \cite{pournin}, and Cleary and Maio \cite{badconflicts} which are examples of difficult tree pairs of increasingly large sizes but though there are multiple possible examples of increasing size, these numbers do not grown nearly as fast as the set of all possible difficult pairs.

As far as the degree of uniformity, computations for small $n$ with exhaustive coverage show essentially complete coverage of difficult instances with factors in the range of 2 between the first and third quartiles of number of instances and a factor of about 7 between the most commonly and least commonly generated.


\bibliographystyle{plain}
\bibliography{rotbib}
\end{document}